\documentclass[11pt]{article}

\usepackage[usenames,dvipsnames]{xcolor}
\usepackage[colorlinks,citecolor=blue,linkcolor=BrickRed]{hyperref}
\usepackage{makeidx} 
\usepackage{graphicx,tipa}
\usepackage{arcs,lmodern,fix-cm}
\usepackage{times}
\usepackage{amsfonts,latexsym,graphicx,epsfig,amssymb,color}
\usepackage{mathdots,amsmath,amsthm,amstext,setspace}
\usepackage{amsmath,amstext,amsthm,url,setspace, enumerate}
\usepackage{slashbox,multirow}
\usepackage{rotating}
\usepackage{verbatim}

\usepackage{float}
\usepackage{graphicx}   
\usepackage{xcolor}
\usepackage{microtype}%if unwanted, comment out or use option "draft"
\usepackage{soul}
\usepackage{algorithm2e}
\usepackage{tabularx}
\usepackage{placeins}
\usepackage{mathtools}

\SetKwBlock{Repeat}{repeat}{end}

\newtheorem{theorem}{Theorem}[section]

\newtheorem{lemma}[theorem]{Lemma}

\theoremstyle{definition}

\newtheorem{definition}[theorem]{Definition}

\newcommand{\reals}{\mathbb{R}}

\newcommand{\ignore}[1]{}

\newcommand{\sinn}[1]{\sin\left({#1}\right)}
\newcommand{\coss}[1]{\cos\left({#1}\right)}

\newcommand{\cycle}[1]{\left(\coss{#1},\sinn{#1}\right)}

\newcommand{\shoreline}[1]{\textsc{Shoreline}$_{#1}$}

\begin{document}
\title{\bf 
Lower Bounds for Shoreline Searching \\
with 2 or More Robots
\footnote{
This is an updated version of the paper with the same title which will appear in the proceedings of the 23rd International Conference on Principles of Distributed Systems (OPODIS’19) Neuch\^atel, Switzerland, July 17-19, 2019.}
}

\author{
Sumi Acharjee\footnotemark[2] %~\footnotemark[5] 
\and
Konstantinos Georgiou\footnotemark[2] %~\footnotemark[5] 
\and
Somnath Kundu\footnotemark[2] %~\footnotemark[6]
\and 
Akshaya Srinivasan\footnotemark[3] %~\footnotemark[5]
}

\def\thefootnote{\fnsymbol{footnote}}
%\footnotetext[5]{Research supported in part by NSERC of Canada.}
%\footnotetext[6]{Research supported in part by the Ontario Graduate Scholarship (OGS) Program.}
\footnotetext[2]{
Dept. of Mathematics, Ryerson University, Toronto, Canada, \\ \texttt{
$\{$sumi.acharjee,konstantinos,somnath.kundu$\}$@ryerson.ca}
}
\footnotetext[3]{
Dept. of Computer Science \& Engineering,  National Institute of Technology, Tiruchirappalli, India,
\texttt{akshaya.kms@gmail.com}
}

\maketitle

\begin{abstract}
Searching for a line on the plane with $n$ unit speed robots is a classic online problem that dates back to the 50's, and for which competitive ratio upper bounds are known for every $n\geq 1$, see~\cite{baeza1995parallel}.
In this work we improve the best lower bound known for $n=2$ robots~\cite{baeza1995parallel} from 1.5993 to 3. Moreover we prove that the competitive ratio is at least $\sqrt{3}$ for $n=3$ robots, and at least $1/\coss{\pi/n}$ for $n\geq 4$ robots.
Our lower bounds match the best upper bounds known for $n\geq 4$, hence resolving these cases. 
To the best of our knowledge, these are the first lower bounds proven for the cases $n\geq 3$ of this several decades old problem.

\vspace{0.5cm}
\end{abstract}

\section{Introduction}

Searching for a shoreline is the problem in which a number of identical unit speed searchers, starting from the same point on the plane, need to agree on trajectories so as to hit (eventually) any line on the plane. The underlying optimization problem asks for fixed  trajectories, one for each searcher, so as to minimize the worst case relative time untill the first searcher hits the line, i.e. the time untill the line is found divided by the distance of the line to the origin. This two-dimensional search-type problem has a long history, and conjectured optimal strategies have been proposed for every $n\geq 1$ (see Section~\ref{sec: related work} for detailed discussion), where $n$ is the number of searchers (robots). 
Similarly to the much easier one dimensional analog of the problem, known as the cow-path problem, showing competitive ratio lower bounds for the problem has been a much more challenging task. Indeed, for the shoreline problem, very weak unconditional lower bounds are known for $n=1$, while the only non-trivial lower bounds known for other values of $n$ is that for $n=2$.

In this work we improve the state-of-the-art when it comes to competitive ratio lower bounds for searching for a shoreline with $n\geq 2$ robots. In particular, we improve the best lower bound known for $n=2$ robots, from 1.5993 to 3. Then, we prove the first lower bounds for $n\geq 3$ robots. More specifically, we show a lower bound of $\sqrt{3}$ for 3 robots, and $1/\coss{\pi/n}$ for $n\geq 4$ robots, matching this way the best upper bound known for the latter case.

\subsection{Related Work}
\label{sec: related work}

Theory of search has a long history that dates back to the 50's, see \cite{beck1964linear,bellman1963optimal}.
In one of the simplest continuous problems, a unit speed robot is moving on an infinite line, and its goal is to hit every point (bounded away from the origin) within bounded relative time. The problem, now known as linear-search or cow-path, was restudied by the computer science community in the late 80's in~\cite{baeza1988searching},
and became so fruitful that numerous variations emerged with challenging and particularly interesting algorithmic problems. Indeed over the decades, accumulated results were summarized in a number of interesting surveys, e.g. 
\cite{benkoski1991survey,CGK19search,dobbie1968survey,gal2010search}. % listed in chronological order.
Moreover, the underlying mathematical theory became rich enough to give rise to a number of related books, with \cite{alpern2003theory,alpern2013search,alpern2006theory} being the most relevant and influential. 

Among the numerous variations/generalizations of the cow-path problem, the current work focuses on its 2-dimensional analog, that we call the \shoreline{n}\ problem, in which $n$ robots are searching in parallel on the plane for a line. As it is outside the scope of this paper to do a thorough literature review on search-type problems, we refer the reader to the aforementioned surveys and books for all remotely related results, and we focus here on the literature closely related to the shoreline problem, i.e. to 2-dimensional search problems with $n\geq 1$ robots. 
The language that we adopt for quantifying algorithms' performance is that of competitive analysis, e.g. see~\cite{borodin2005online}. In particular, we think of our problem as $n$ robots, starting from the origin, that are solving an online problem in which a line is placed at an unknown location $\delta$ away from the origin, where in particular $\delta$ is unknown (but it is bounded away from 0, and that bound is known).
The goal of the search is to minimize the relative worst case search time, i.e. the time untill the line is found by any robot divided by $\delta$, over all possible placements of lines and over all $\delta$. The best possible relative time is known as the \textit{competitive ratio} of the problem, and can be thought as the best worst case relative performance of an online algorithm (that does not know the input) compared to the performance of the best offline algorithm (that knows the input).

Searching for a (shore)line with 1 robot, without any knowledge of its distance to the origin, was first proposed in~\cite{baeza1988searching}, and a number of improvements were proposed for parallel search 
in~\cite{baeza1997searching,baeza1995parallel,baezayates1993searching,jez2009two}, 
i.e. for $n\geq 2$ robots and a number of variations. 
The best algorithm known for \shoreline{1} is a logarithmic spiral search that has competitive ratio $13.81$~\cite{baeza1988searching}. Notably, the only unconditional lower bound for the problem is that the competitive ratio is at least $6.3972$~\cite{baeza1995parallel} which also holds true if the online algorithm knows, a priori, the distance of the line to the origin. Only assuming a cyclic-type trajectory, the competitive ratio is provably at least $12.5385$~\cite{langetepe2012searching}.

The overall picture for searching with $n\geq 2$ robots for a line (without any knowledge about the hidden line) is much more blurry. For $n=2$, a double logarithmic spiral, in which the origin lies always in the middle of the locations of the robots is known to induce competitive ratio $5.2644$~\cite{baeza1995parallel}. The only lower bound to the problem is due to the variation in which the distance is known, and it is 1.5993~\cite{baeza1995parallel}. For $n\geq 3$, the natural algorithm of~\cite{baeza1995parallel} makes robots move along rays, splitting the plane evenly, and induces competitive ratio at most $1/\coss{\pi/n}$. To the best of our knowledge, no competitive ratio lower bounds have been reported for problems $n\geq 3$.  

Some relevant variations to our problem are those in which partial information, e.g. the slope or the distance to the origin, is known regarding the hidden line. All results in this paragraph refer to searching with 1 robot. 
When both distance and slope are known, the best possible algorithm has competitive ratio 3. When the distance is known, and the line is axis parallel, then the best competitive ratio is $3\sqrt{2}$~\cite{baeza1988searching}. 
If only the distance is known, \cite{isbell1957optimal} gives the best deterministic online algorithm with competitive ratio $6.39$. Randomized online algorithms for the same problem were proposed in~\cite{gluss1961alternative,gluss1961minimax}.
The problem in which the slope is known is the traditional cow-path problem with best possible competitive ratio 9 and was studied in \cite{baeza1988searching,gal2010search}.
When the line is known to be axis parallel, then \cite{baeza1988searching} gives an upper bound of 13.02, which was improved to $12.5406$~\cite{jez2009two} and then to $12.5385$~\cite{langetepe2012searching}, the latter shown to be optimal among  cyclic-type trajectories. 
As stated previously, when no information is known the best upper and lower bounds known are $13.81$ and (conditionally to cyclic-trajectories) $12.5385$, respectively, due to \cite{baeza1988searching} (technical report~\cite{finch2005searching} has a nice exposition of the same upper bound with all mathematical derivations).

Two-dimensional search problems have been considered beyond line searching. Indeed, \cite{gluss1961minimax} considered the problem of searching for a circle. In 2010, Langetepe~\cite{langetepe2010optimality} showed that spiral search is optimal for 2 dimensional search by one robot, assuming that all points that are convex combinations of robot's trajectory and the origin are seen/discovered. 
The same problem with more robots was studied in \cite{fricke2016distributed}.
Papers
\cite{Emekicalp2014,emek2015many,LangnerKUW15,Lenzen2014,LS01}
consider parallel search on the grid with bounded memory robots. 
\cite{bouchard2018deterministic} and \cite{pelc2018reaching}  considered other variations of the problem of searching for a point in the plane, while~\cite{pelc2018information} considered searching for a point within a geometric terrain. Finally, the very recent~\cite{pelc2019cost} studied cost/information trade-offs for searching in the plane for a point.

\subsection{Problem Definition and Summary of Known and New Results}
\label{sec: def}

We begin this section with a formal description of the two-dimensional search problem, first considered in~\cite{baeza1988searching}. 

\begin{definition}[\shoreline{n}: Searching for a Shoreline with $n$ Robots] ~\\
$n$ unit speed robots start from the origin of the plane. 
Feasible solutions to the problem are robots' trajectories $\mathcal F_n$, such that for every line $\ell$ of the plane, there exists at least one robot's trajectory intersecting $\ell$. The time $T_{\mathcal{F}_n}(\ell)$ by which $\ell$ is hit for the first time is the \textit{search completion time}. If $\delta(\ell)$ represents the distance of $\ell$ to the origin, the objective of \shoreline{n} is to find trajectories $\mathcal F_n$ so as to minimize the search competitive ratio of $\mathcal F_n$ defined as
\begin{equation}
\label{equa: search comp ratio}
CR(\mathcal F_n):= \sup_{\ell} \frac{T_{\mathcal F_n}(\ell)}{\delta(\ell)}.
\end{equation}
The best possible search competitive ratio $\inf_{\mathcal F_n} CR(\mathcal F_n)$ will be denoted by $\mathcal S_n$. 
\end{definition}

In order to avoid degenerate cases, especially when $n=1$, the supremum of~\eqref{equa: search comp ratio} can be restricted to lines $\ell$ for which $\delta(\ell)\geq \epsilon$, for some $\epsilon >0$ that is known to algorithm $\mathcal F_n$. Also, for the rest of the paper, and when it is convenient, we will study \shoreline{n} from the perspective of analytic geometry, that is robots will start from the origin of the Cartesian plane, and trajectories will be analytic curves in the plane.

Problem~\shoreline{n}, and variations of it, have been studied as early as in the late 50's for $n=1$ and in the 80's for $n\geq 2$. 
An upper bound to $\mathcal S_1$ of 13.81 was reported in~\cite{baeza1988searching}
and a lower bound of $12.5385$ in~\cite{langetepe2012searching}, assuming that the solution trajectory is of spiral-type. 
The only unconditional lower bound to $\mathcal S_1$ is that of 6.3972, see~\cite{isbell1957optimal}, and refers to instances/lines with known distance to the origin, and hence apply to instances with unknown distance as well. 

Notably, for the case $n\geq 2$ almost no lower bounds are known and are restricted to instances with known distance to the origin.
Indeed, \cite{baeza1995parallel} reports that $\mathcal S_2 \geq 1.5993$ and the trivial $\mathcal S_n \geq 1$, for $n\geq 4$. When it comes to upper bounds, a double-spiral trajectory performed by two robots ensures that $\mathcal S_2 \leq 5.2644$, see~\cite{baeza1995parallel}. For the case $n\geq 3$, \cite{baeza1995parallel} proposes the following ray-type algorithm: for $i=1,\ldots,n$, robot $i$, searches along the ray with direction $\cycle{(i-1)\phi }$, where $\phi=2\pi/n$. It is an easy exercise to show that this algorithm witnesses that $\mathcal S_n \leq \tfrac{1}{\coss{\pi/n}}$, again for $n\geq 3$. 

\subsubsection{Organization of the paper} Our main contributions pertain to \textit{new} lower bounds for $\mathcal S_n$, when $n\geq 2$, which in particular for $n\geq 4$ are tight. 
More specifically, we show 
that $\mathcal S_2 \geq 3$ (see Theorem~\ref{thm: n=2 lower bound} in Section~\ref{sec: n=2 lower bound}), 
that $\mathcal S_3 \geq \sqrt{3}$ (see Theorem~\ref{thm: n=3 lower bound} in Section~\ref{sec: n=3 lower bound}), 
and that $\mathcal S_n \geq \tfrac{1}{\coss{\pi/n}}$ (see Theorem~\ref{thm: n>=4 lower bound} in Section~\ref{sec: n>=4 lower bound}) for $n\geq 4$. The exposition of the results is in reverse order due to the nature of our arguments. 
Combined with the known upper bounds discussed above, our results imply the following state-of-the-art regarding problem \shoreline{n}, when $n\geq 2$. 
\begin{theorem}
For the best possible competitive ratio $\mathcal S_n$ for \shoreline{n} we have that:\\
$3\leq \mathcal S_2 \leq 5.2644$, \\
$\sqrt{3} \leq \mathcal S_3 \leq 2$, \\
$\mathcal S_n = 1/\coss{\pi/n}$, for all $n\geq 4$.
\end{theorem}

\section{Lower Bounds for $n\geq 4$ Robots}
\label{sec: n>=4 lower bound}

This section is devoted to proving the following theorem. 
\begin{theorem}
\label{thm: n>=4 lower bound}
For all $n\geq 4$, we have $\mathcal S_n \geq 1/\coss{\pi/n}$.
\end{theorem}

The proof of the theorem above is split in a number of lemmata. 
First, we show in Lemma~\ref{lem: OMB} that under certain conditions, optimal robots' moves are along straight lines.

\begin{lemma}
\label{lem: OMB}
Consider right triangle $OMB$ with $\angle{BOM}\leq \pi/4$ and $\angle{OMB}=\pi/2$. 
Then for every point $K$ in the line segment $MB$ and every point $L$ in the line segment $OB$ we have that $OK+KL\geq OB$, and equality is satisfied only when all points $K,B,L$ coincide. 
\end{lemma}

\begin{proof}
Consider an arbitrary point $L$ in the interior of the line segment $OB$ of right triangle $OMB$, see also Figure~\ref{fig: OMB}.
By the law of reflection, among all points $K$ in the interior of segment $BM$, the one that minimizes $OK+KL$ is the point for which 
$\angle{MKO}=\angle{LKB}$.
%, which angle we also denote by $\omega$. 
We fix such a point $K$, and clearly it is enough to show that $OK+KL>OB$. 

We consider the reflection of the triangle $\Delta OMB$ about axis $BM$. Note that points $O', K, L$ are co-linear, see also Figure~\ref{fig: OMB}.
\begin{figure}[h!]
\begin{center}
 \includegraphics[width=7cm]{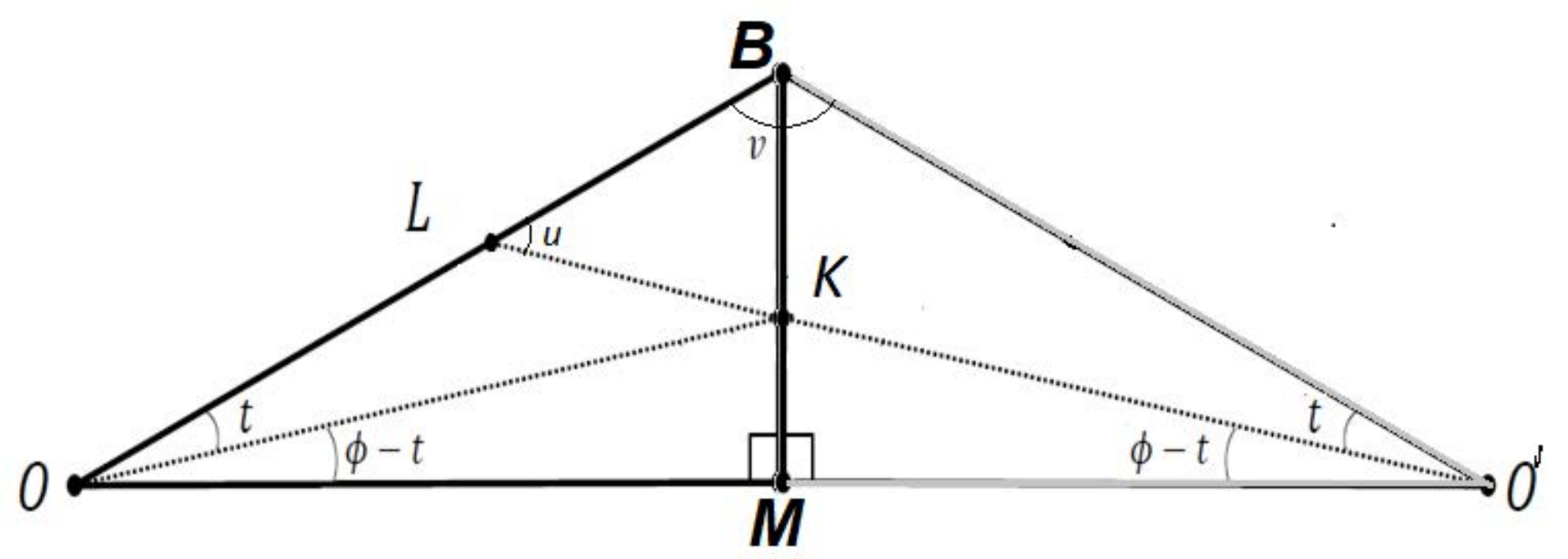}
\end{center}
 \caption{Triangle $OMB$ of the proof of Lemma~\ref{lem: OMB}.}
 \label{fig: OMB}
\end{figure}
For notational convenience, we introduce abbreviations
$\phi:=\angle BOM = \angle BO'M$, 
$t:= \angle BOK = \angle BO'K  $, 
$v:= \angle OBO' = \pi - 2 \cdot \phi$,
$u :=\angle BLO'= \pi - v - t = 2 \cdot \phi - t$.
Therefore $ v - u = \pi + t - 4 \cdot \phi$. 
But since $\phi \leq \frac{\pi}{4}$,  we conclude that $v - u \geq t \geq 0 $. So, in triangle $BLO'$ we have $\angle LBO' \geq \angle BLO'$. Therefore, we conclude that 
$O'L \geq O'B$. By symmetry, we also deduce that $OK + KL \geq OB$, as wanted. 
\end{proof}

Next, using the lemma above, we show that under certain conditions, there are lines that optimal search trajectories cannot have discovered within certain time bounds.

\begin{lemma}
\label{lem: OAB trajectory}
Consider trajectories $\mathcal F_n$ for problem \shoreline{n}, where robots start from origin $O$, and fix time $d>0$. Consider a cone $\mathcal C$ of angle $2\phi$, where $0<\phi \leq \pi/4$, and centered at $O$. Let also $A,B$ be two points at the two extreme rays of the cone, such that $OA=OB=d+\epsilon$, for some $\epsilon>0$. If at time $d$, there is no robot within the cone $\mathcal C$, then the line passing through points $A,B$ could not have been intersected by the trajectory of any robot. 
\end{lemma}

\begin{proof}
The proof is by contradiction, so we assume that a robot's trajectory has intersected line $\ell$ passing through $A,B$ and that the robot is outside cone $\mathcal C$. Since we posed the execution of the algorithm at time $d$, and since $OA=OB>d$, robot's trajectory could not have intersected the extreme rays of cone $\mathcal C$ further than points $A,B$. Since the robot is outside the cone, 
the trajectory of the robot must have intersected (for the first time) segment $AB$ in some interior point $K$ and then segment $OB$ (or $OA$) in some interior point $L$ (after hitting line $\ell$), see also Figure~\ref{fig: OABtrajectory}. Also, without loss of generality, $K$ is closer to $B$ than from $A$. 
\begin{figure}[h!]
\begin{center}
 \includegraphics[width=7cm]{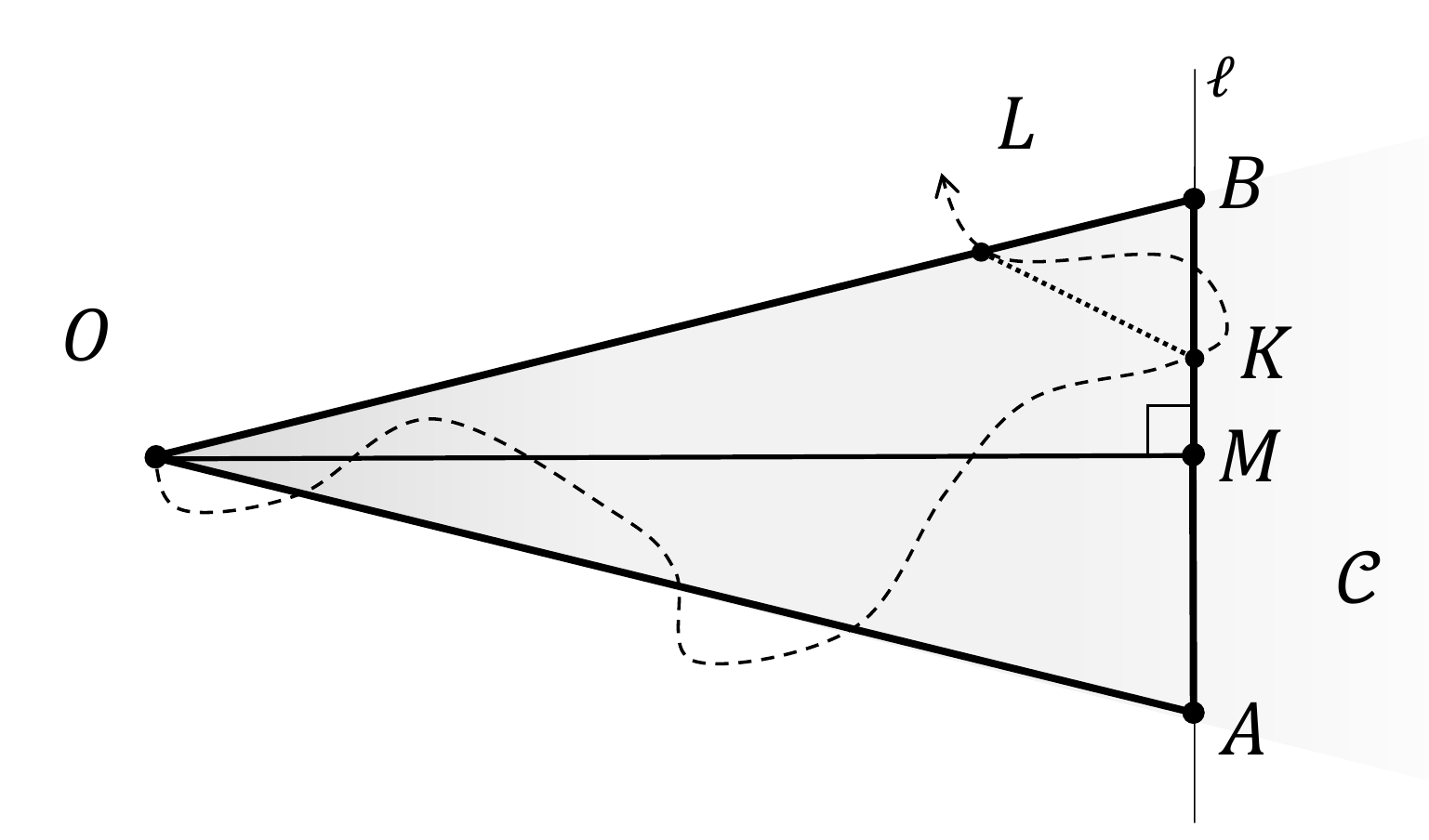}
\end{center}
 \caption{Cone $\mathcal C$ of the proof of Lemma~\ref{lem: OAB trajectory}. Robot's trajectory is depicted as the curved dotted line.}
 \label{fig: OABtrajectory}
\end{figure}
Since robot's trajectory takes place in the Euclidean space, and robot has unit speed, the time for such a trajectory to be realizable is at least $OK+OL$. Consider then the projection $M$ of the origin $O$ onto line segment $AB$, and observe that $\angle{BOM}=\phi\leq \pi/4$, since triangle $BOA$ is isosceles. But then, Lemma~\ref{lem: OMB} applies according to which the time that has passed is at least $OK+OL\geq OB=d+\epsilon>d$, a contradiction. 
\end{proof}

Next we quantify a lower bound to the competitive ratio of search algorithms in which robots exhibit a certain property. 

\begin{lemma}
\label{lem: outside cone implies cr}
Consider trajectories $\mathcal F_n$ for problem \shoreline{n}, where robots start from origin $O$.
If there is a cone of angle $2\phi$ centered at the origin, where $0<\phi \leq \pi/4$, within which there is no robot at an arbitrary time (or a robot lies at the origin), then $CR(\mathcal F_n) \geq 1/\coss{\phi}$.
\end{lemma}

\begin{proof}
Consider a time $d>0$, and a cone centered at the origin of angle $2\phi$, such that no robot lies within the cone. 
Consider points $A,B$ on the extreme rays of the cone at distance $d+\epsilon$ from the origin, for an arbitrary small $\epsilon>0$. 
By Lemma~\ref{lem: OAB trajectory}, no robot could have discovered the line $\ell$ passing through points $A,B$. Since time $d$ has already passed, we conclude that the search completion time satisfies $T_{\mathcal{F}_n}(\ell)>d$. At the same time, triangle $OAB$ is isosceles, and so the distance of $\ell$ and the origin is $\delta(\ell)=OB\cdot \coss{\phi} = (d+\epsilon)\coss{\phi}$. We conclude that 
$$
CR(\mathcal F_n)
\geq \sup_{\epsilon>0} \frac{d}{(d+\epsilon)\coss{\phi}}=\frac1{\coss{\phi}}.
$$
Finally consider the case that the only robot within the cone lies at the origin. That robot cannot reach the line earlier than $d+(d+\epsilon)\coss{\phi}$, hence the same bound for the competitive ratio holds. 
\end{proof}

Now we are ready to prove Theorem~\ref{thm: n>=4 lower bound}.

\begin{proof}[Proof of Theorem~\ref{thm: n>=4 lower bound}]
Fix $n$, and consider trajectories $\mathcal F_n$ for problem \shoreline{n}, where robots start from origin $O$.
Let robots move for an arbitrary time $d>0$. 
If all robots lie at the origin, then clearly the competitive ratio is unbounded. 

Otherwise, consider a cone of arbitrary small angle $2\gamma=o(1/n)$ centered at the origin. 
We rotate the cone untill at least one robot (note there exists at least one not in the origin) lies strictly within this small cone. 
Then we cover the plane by concatenating, in an alternate fashion and clockwise to the existing small cone, $n$ many cones centered at the origin of angle $2\pi/n-2\gamma$, and $n-1$ more cones centered at the origin of angle $2\gamma$. 

Note, there are $n$ ``small'' cones of angle $2\gamma$, one of which strictly contains a robot, and $n$ many ``large'' cones of angle $2\pi/n-2\gamma$ and hence one of which, call it $\mathcal C$, does not contain any robot, unless a robot is at the origin. 
But then, Lemma~\ref{lem: outside cone implies cr} applies with $\phi=2\pi/n-2\gamma$, for any $\gamma$. 
That is, for the arbitrary trajectories $\mathcal F_n$, and for every $\gamma>0$ we have that 
$
CR(\mathcal F_n) \geq \frac1{\coss{\pi/n-\gamma}},
$
hence $\mathcal S_n \geq \frac1{\coss{\pi/n}}$ as wanted. 
\end{proof}

\section{Lower Bound for $3$ Robots}
\label{sec: n=3 lower bound}

In this section we prove the following theorem. 

\begin{theorem}
\label{thm: n=3 lower bound}
$\mathcal S_3 \geq \sqrt{3}$.
\end{theorem}

Notably, the achieved lower bound does not match the best upper bound known for \shoreline{3}. 
In particular, the lower bound arguments of Section~\ref{sec: n>=4 lower bound} fail for \shoreline{n}, when $n<4$. 
Indeed, the crux of the previous argument is that robots should lie at the boundary (extreme rays) of cones, centered at the origin and of angles $2\phi_n := 2\pi/n$. If robots are given, say, time 1 to execute their trajectories, then there are special lines which are $\coss{\phi_n}+\epsilon$ away from the origin, that could not have been visited by any robot, because otherwise the robots would not have enough time to leave from some cones. The crucial necessary condition of the previous statement is that $\phi_n\leq \pi/4$, which of course holds when $n\geq 4$. 
In the case of $n=3$, robots can visit these special lines in time less than 1, still leaving the cones of angle $2\pi/3$, hence making the argument invalid. However, the robots would still need a significant amount of time (bounded away from 0) to achieve the same task, hence placing the special lines sufficiently further away would allow the argument to go through. 

The paragraph above gives the high level idea of the proof of Theorem~\ref{thm: n=3 lower bound}, and also explains why we presented first, in Section~\ref{sec: n>=4 lower bound}, the lower bounds for $n\geq 4$ robots. 
The next lemma establishes a lower bound for the time that robots need, in \shoreline{3}, to discover the special lines that were used for the lower bounds to $\mathcal S_n$, for $n\geq 4$. 
\begin{lemma}
\label{lem: scaling}
Consider a cone of angle $2\pi/3$ centered at the origin $O$, along with two points $A,B$ on its extreme rays at distance 1 from $O$. 
Then, a unit speed robot starting from $O$ requires time at least $\sqrt{3}/2$ to visit the line passing through $A,B$ and leave the cone. 
\end{lemma}

\begin{proof}
Consider the projection $M$ of $O$ onto the line $\ell_{AB}$ passing through $A,B$.
We calculate the shortest trajectory starting from $O$, visiting an arbitrary point $K$ of line segment $MB$ and leaving the cone from an arbitrary point $L$ of line segment $OB$, see also Figure~\ref{fig: scaling}.
\begin{figure}[h!]
\begin{center}
 \includegraphics[width=7cm]{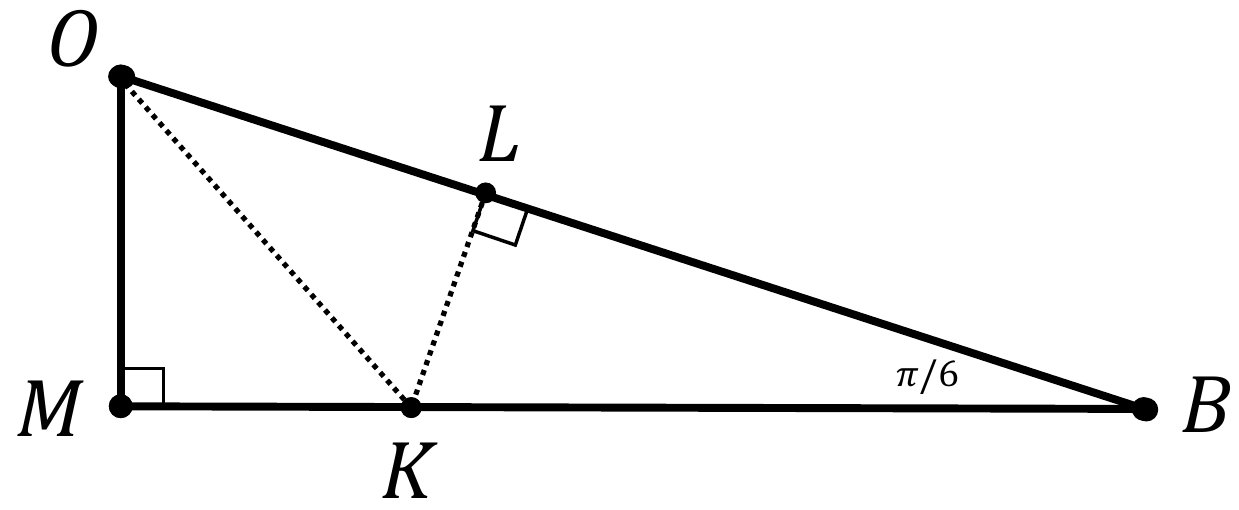}
\end{center}
 \caption{The shortest trajectory, starting from $O$, visiting line segment $MB$ at point $K$ and leaving the cone from point $L$ from an extreme ray, is depicted with dotted lines. }
 \label{fig: scaling}
\end{figure}

For convenience, we introduce a coordinate system centered at $M$, so that $O=(0,1/2)$ and $B=(\sqrt{3}/2,0)$ (recall that $OB=1$). 
The arbitrary point $K$ on $MB$ is a convex combination of points $M,B$, and hence has coordinates $K=\lambda (\sqrt{3}/2,0)$, for some $\lambda \in [0,1]$. Note that $OK=\sqrt{\tfrac34\lambda^2+\tfrac14}=\tfrac12\sqrt{3\lambda^2+1}$.
Given that $K$ is chosen, the shortest path for leaving the cone is clearly the distance $d(K,\ell_{AB})$ between $K$ and $\ell_{OB}$ passing through $O,B$. It is easy to see that $\ell_{OB}$ is described as $y+\sqrt{3}/3x-1/2=0$, hence, 
$$
d(K,\ell_{OB}) = \frac{
\left|\tfrac{\sqrt3}3\lambda\tfrac{\sqrt3}2-\tfrac12\right|
}{\sqrt{1+\tfrac39}}
=\frac{\sqrt{3}}4(1-\lambda).
$$

We conclude that the shortest path in order to start from $O$ and leave the cone is 
$$
\min_{\lambda\in [0,1]} 
\left\{
OK+
d(K,\ell_{OB})
\right\}
=
\min_{\lambda\in [0,1]} 
\left\{
\frac12\sqrt{3\lambda^2+1}+\frac{\sqrt{3}}4(1-\lambda)
\right\}
$$
We calculate the derivative of the latter function $f(\lambda)$ as $f'(\lambda)=\tfrac{3 \lambda }{2 \sqrt{3 \lambda ^2+1}}-\frac{\sqrt{3}}{4}$, which has a unique root at $\lambda_0=1/3$. Then, we calculate $f''(1/3)=\tfrac{9 \sqrt{3}}{16}$, which shows that $\lambda_0\in [0,1]$ is indeed a minimizer, inducing a trajectory of smallest length $f(1/3)=\sqrt{3}/2$. Since the robot has unit speed, this is also the minimum time needed to reach the cone after visiting line $\ell_{AB}$. 
\end{proof}

We are now ready to prove Theorem~\ref{thm: n=3 lower bound}. 

\begin{proof}[Proof of Theorem~\ref{thm: n=3 lower bound}]
Consider trajectories $\mathcal F_3$ for problem \shoreline{3}, where robots start from origin $O$.
Let robots move for an arbitrary time $d>0$, and consider 3 cones centered at the origin, each of angle $2\pi/3$, covering the entire plane. 
We rotate the cones, untill at least one of the robots lies on the extreme ray of two cones.
As a result, there is a cone $\mathcal C$ such that no robot lies within the interior of the cone. 

Now consider points $A,B$ on the extreme rays of $\mathcal C$ that are $(2/\sqrt{3}+2\epsilon)d$ away from the origin $O$, and let the line passing through them be $\ell_{AB}$. Note that $\ell_{AB}$ is exactly $(1/\sqrt{3}+\epsilon)d$ away from the origin. 

By Lemma~\ref{lem: scaling}, and since no robot lies within $\mathcal C$, no robot could have discovered $\ell_{AB}$ in time $d$, hence the search completion time is at least $d$. Overall, that induces competitive ratio for $\mathcal F_3$ at least $1/(1/\sqrt{3}+\epsilon)$, for every $\epsilon>0$. 
\end{proof}

\section{Lower Bound for $2$ Robots}
\label{sec: n=2 lower bound}

In this section we prove a lower bound for searching for a shoreline with 2 robots. 

\begin{theorem}
\label{thm: n=2 lower bound}
$\mathcal S_2 \geq 3$.
\end{theorem}

In order to prove our theorem, we describe robots' trajectories within time 1. The following function, the boundary of an ellipsoid, will be useful in our calculations
$$
q(x,y,\delta,\theta):=
4
\left( 
\coss{\theta}x + \sinn{\theta}y - h_\delta
\right)
^2 
+
\left( 
-\sinn{\theta}x + \coss{\theta}y
\right)
^2
/ b_\delta^2
-1,
$$
where $h_\delta := \delta/2$ and $b_\delta := \sqrt{(1-\delta^2)/2}$.

\begin{lemma}
\label{lem: ellipses}
Consider an arbitrary algorithm $\mathcal F_2$, and let robots execute it for time 1.
Then, there exist $\epsilon, \delta\in [0,1]$ and $\theta \in [0,\pi]$ so that no point outside the ellipses 
$ q(x,y,\epsilon,0)\leq 0,$ and $q(x,y,\delta,\theta)\leq 0$ has been explored by any robot. 
\end{lemma}

\begin{proof}
Consider an arbitrary algorithm $\mathcal F_2$, and let robots execute it for time 1. Suppose that the locations of robots \#1, \#2 are  $R_1,R_2$ within the unit ball. We claim that the collection of points that each robot could have visited form two ellipses. Indeed, without loss of generality both robots lie in the first two quadrants, i.e. in the non-negative y-axis hyperplane. Also, without loss of generality, the location of $R_1$, exactly at time 1 of the execution of $\mathcal F_2$, equals $(\epsilon, 0)$, for some $\epsilon \in [0,1]$. Recall that robots have speed 1, and they started from the origin $O=(0,0)$. Therefore, all the points $P$ that could have been visited by robot \#1 satisfy $OP + PR_1 \leq 1$. In other words, the boundary of the explored domain of that robot is an ellipse with foci points $O,R_1$ and sum of distances to foci equal to 1. 
Then, all boundary points $(x,y) \in \reals^2$ of the domain that could have been explored by robot \#1 satisfy 
$
4(x-h_\epsilon)^2 + y^2/b_\epsilon^2 =1.
$

Similarly, robot \#2 is, at the same time, at distance $\delta$ from the origin, for some $\delta \in [0,1]$. Arguing as above, the boundary of the explored domain by robot \#2 is again an ellipse. 
So suppose that $R_2= \left( \delta \coss{\theta}, \delta \sinn{\theta}\right)$ for some $\theta \in [0,\pi]$ (that is the line passing through its two foci forms angle $\theta$ with the $x$-axis) Then the boundary of the explored domain by robot \#2 is defined as 
$
4
\left( 
\coss{\theta}x + \sinn{\theta}y - h_\delta
\right)
^2 
+
\left( 
-\sinn{\theta}x + \coss{\theta}y
\right)
^2
/ b_\delta^2
=1
$
(observe that a rotation by angle $-\theta$ gives a formula identical to the one of robot \#1). Finally, note that since robot \#2 lies in the first two quadrants, we must have $\theta \in [0,\pi]$. 
%Proof of claim: Let $(x,y)$ be the points of the new ellipse. Rotating the points by angle $-\theta$ defines the previous ellipse with points $(x',y')$ satisfying the previous equation with $\delta$ instead of $\epsilon$. But then, $R_{-\theta} (x,y)^T = (x',y')^T$, and so the formula follows. 
\end{proof}

The idea behind the proof of  Theorem~\ref{thm: n=2 lower bound} is that if at a certain time, robots \#1, \#2 are at points $R_1= (\epsilon,0)$ and $R_2= \left( \delta \coss{\theta}, \delta \sinn{\theta}\right)$ (for some $\epsilon,\delta \in [0,1]$ and $\theta\in [0,\pi]$), respectively, then they could not have been in any point past the line $y=-1/2$, see Figure~\ref{fig: ellipse example} for an example. 
\begin{figure}[h!]
\begin{center}
 \includegraphics[width=8cm]{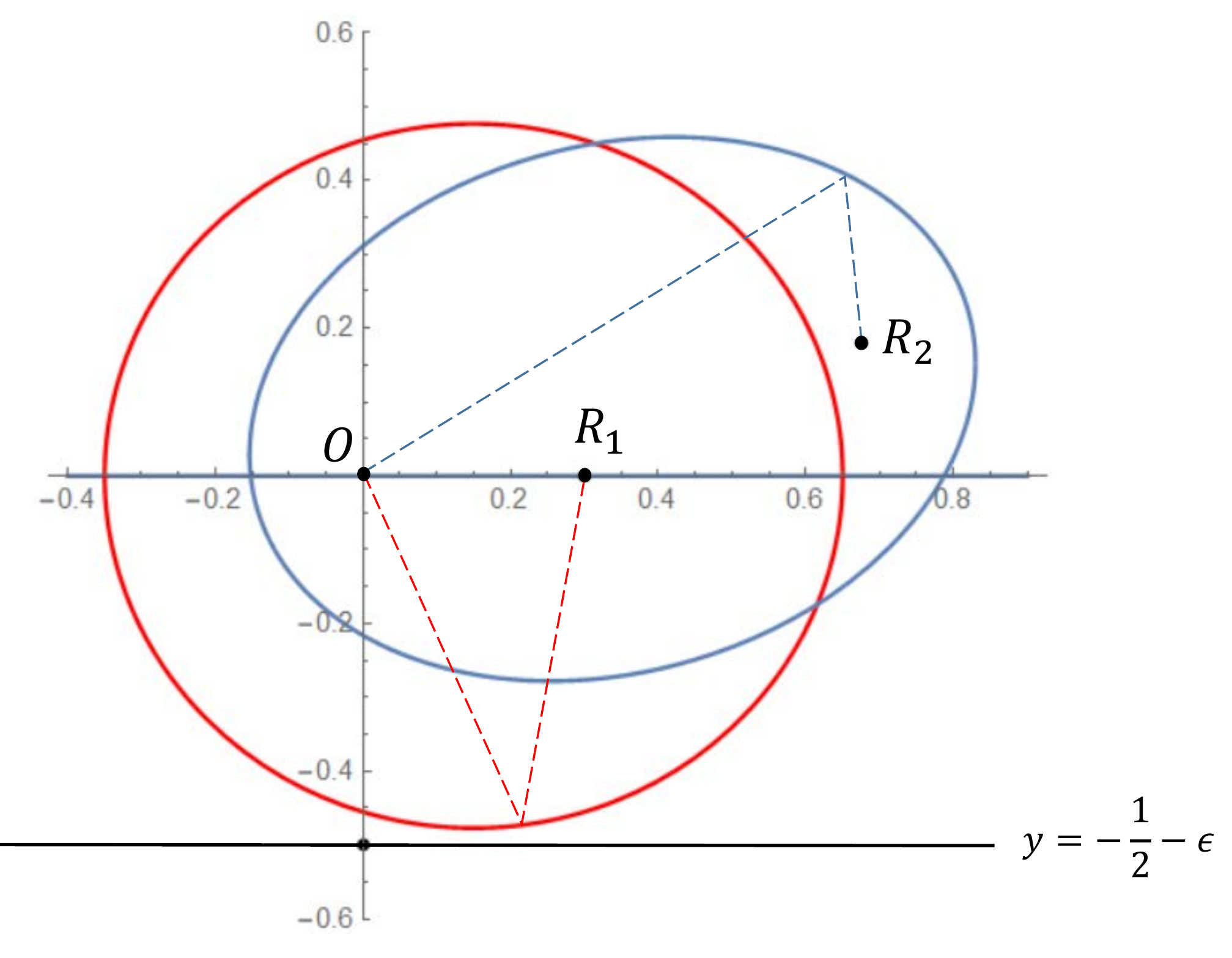}
\end{center}
 \caption{An example of possible robots' placements of an arbitrary search algorithm after time 1. Without loss of generality, robot \#1 lies on the positive x-axis, here depicted as point $R_1=(0.3,0)$. Similarly, robot \#2 lies, without loss of generality, in any of the first two orthants. Here it is depicted as point $R_2=(0.7\coss{\pi/12}, 0.7\sinn{\pi/12})$. Both robots have started from the origin, indicated by $O$. Possible points that robots \#1,\#2 have visited are depicted by the red and blue ellipses, respectively. Dotted straight trajectories show possible robots' movements from the origin, then to arbitrary points on the boundaries of the ellipses, and then to points $R_1,R_2$.}
 \label{fig: ellipse example}
\end{figure}
We are now ready to prove Theorem~\ref{thm: n=2 lower bound}.

\begin{proof}[Proof of Theorem~\ref{thm: n=2 lower bound}]
Consider an arbitrary search algorithm $\mathcal F_2$. According to Lemma~\ref{lem: ellipses}, there are two ellipses that define all points that could have been explored by any of the robots. 
Our main claim is that neither of the two robots could have hit line $\ell:~y=-1/2-\zeta$, where $\zeta>0$ is arbitrarily small. For that, all we need is to show that none of the equations defining any of the two ellipses have any common point with $y=-1/2-\zeta$. To that end, we show that equation $q(x,-1/2-\zeta,\delta,\theta)=0$ has no real root, when $\delta\in [0,1], \theta \in [0,\pi]$ and $\zeta>0$ is sufficiently small (that would also imply the same for the first ellipse). Indeed, we compute the discriminant of the degree 2 polynomial $q(x,-1/2-\zeta,\delta,\theta)$ (in $x$) which equals 
\begin{align*}
-\frac{16 \left(\delta ^2+2 \delta  (2 \zeta +1) \sin (\theta )+4 \zeta  (\zeta +1)\right)}{1-\delta ^2}
&\leq 
-\frac{16 \left(\delta ^2+4 \zeta  (\zeta +1)\right)}{1-\delta ^2}
\end{align*}
Since $\zeta>0$ and arbitrarily small, the latter expression is maximized for $\delta=0$ and becomes $-64 \left(\zeta ^2+\zeta \right)$ which is negative for all small enough $\zeta>0$. 

%We are now ready to conclude the proof of Theorem~\ref{thm: n=2 lower bound}. 
Note that the closest robot to line $\ell$ is robot \#1 (or robot \# 2 too, when $\theta=0,\pi$), and its distance to that line equals $1/2+\zeta$. Since time 1 has already passed, the search completion time is at least $3/2+\zeta$. At the same time, the optimal offline solution equals $1/2+\zeta$. Hence, the competitive ratio of the arbitrary search algorithm $\mathcal F_2$ is at least 
$$
\sup_{\zeta>0}\frac{3/2+\zeta}{1/2+\zeta}=3. 
$$
\end{proof}

\section{Open Problems}
We studied the problem of searching for a shoreline with $n$ robots, and in particular we gave strong lower bounds when $n\geq 2$. Our results are tight when $n\geq 4$, completely resolving these cases. The cases $n=2,3$ as well as the case $n=1$, which is not addressed in this work, remain open. It is plausible that the best algorithms known  when $n=1,2$ are indeed optimal, even though a proof seems to be particularly challenging. The case of $n=3$ seems to be the most interesting since the upper bound is provided by the same algorithm as for the cases $n\geq 4$, still our argument that shows optimality for the latter cases fails to be tight when $n=3$. 
Finally, a number of variations of the shoreline problem remain open. These include the cases of different robots specs, e.g. speeds, the possibility of faulty robots, different termination criteria, e.g. evacuation or rendezvous instead of search, different measures of efficiency, e.g. average-case worst-case tradeoffs, etc.

\section*{Acknowledgments}

This research was supported by NSERC Discovery and MITACS Globalinks grants. 

\bibliographystyle{plain}
\bibliography{BiblioSearch}
%\bibliographystyle{abbrvnat}
%\bibliography{}

\end{document}